\documentclass[thm-restate]{lipics-v2021}
\usepackage[utf8]{inputenc}
\usepackage{url}
\usepackage{algorithm}
\usepackage{algpseudocode}
\usepackage{amsmath}
\usepackage{xcolor}

\hideLIPIcs
\nolinenumbers

\title{Faster Fréchet Distance under Transformations}

\date{}

\acknowledgements{}

\author{Kevin Buchin}{Technical University of Dortmund, Germany}{kevin.buchin@tu-dortmund.de}{https://orcid.org/0000-0002-3022-7877}{}
\author{Maike Buchin}{Ruhr University Bochum, Germany}{maike.buchin@rub.de}{https://orcid.org/0000-0002-3446-4343}{}
\author{Zijin Huang}{University of Sydney, Australia}{huang.zi.jin24@gmail.com}{https://orcid.org/0000-0003-3417-5303}{}
\author{André Nusser}{Université Côte d'Azur, CNRS, Inria, France}{andre.nusser@cnrs.fr}{https://orcid.org/0000-0002-6349-869X}{This work was supported by the French government through the France 2030 investment plan managed by the National Research Agency (ANR), as part of the Initiative of Excellence of Université Côte d'Azur under reference number ANR-15-IDEX-01.}
\author{Sampson Wong}{BARC, University of Copenhagen, Denmark}{sawo@di.ku.dk}{https://orcid.org/0000-0003-3803-3804}{}

\authorrunning{K. Buchin, M. Buchin, Z. Huang, A. Nusser, S. Wong} 

\Copyright{Kevin Buchin, Maike Buchin, Zijin Huang, André Nusser, Sampson Wong}

\newtheorem{fact}[theorem]{Fact}

\DeclareUnicodeCharacter{2212}{-}

\theoremstyle{remark}

\newcommand{\abse}[1]{\lVert #1 \rVert}
\newcommand{\abs}[1]{|#1|}

\makeatletter
\def\NAT@spacechar{~}\makeatother

\usepackage{thmtools}
\usepackage{proof-at-the-end}

\newcommand{\Frechet}{Fréchet }
\newcommand{\Oh}{\mathcal{O}}
\newcommand{\tOh}{\tilde{\mathcal{O}}}

\newcommand{\boundary}[1]{\partial #1}

\newcommand{\weight}[1]{\normalfont \textbf{w}(#1)}

\newcommand{\disk}[1]{D(#1)}

\newcommand{\graph}[0]{\mathcal{G}}
\newcommand{\edges}[0]{\mathcal{E}}
\newcommand{\vertices}[0]{\mathcal{V}}

\newcommand{\df}[1]{d_{\mathcal{F}}(#1)}
\newcommand{\fd}[2]{\mathcal{D}_{#2}(#1)}
\newcommand{\D}{\mathcal{D}}

\newcommand{\reals}[0]{\mathbb{R}}

\newcommand{\doublequote}[1]{``#1''}

\newcommand{\vertex}[1]{\normalfont \textbf{v}(#1)}

\newcommand{\fsgraph}[0]{\graph^{f}}
\newcommand{\cell}[1]{C_{#1}}

\newcommand{\ggraph}[0]{\graph^g}
\newcommand{\row}[1]{R(#1)}

\newcommand{\freespace}[2]{\mathcal{F}_{#2}(#1)}

\newcommand{\trans}[0]{\mathcal{T}}
\newcommand{\vvetrans}[0]{\mathcal{T}^{vve}_\delta}
\newcommand{\vetrans}[0]{\mathcal{T}^{ve}_\delta}
\newcommand{\cseg}[1]{\overline{#1}}

\begin{document}
\maketitle

\begin{abstract}
    We study the problem of computing the Fréchet distance between two polygonal curves under transformations. First, we consider translations in the Euclidean plane. Given two curves $\pi$ and $\sigma$ of total complexity $n$ and a threshold $\delta \geq 0$, we present an $\tilde{\mathcal{O}}(n^{7 + \frac{1}{3}})$ time algorithm to determine whether there exists a translation $t \in \mathbb{R}^2$ such that the Fréchet distance between $\pi$ and $\sigma + t$ is at most~$\delta$. This improves on the previous best result, which is an $\mathcal{O}(n^8)$ time algorithm.

    We then generalize this result to any class of rationally parameterized transformations, which includes translation, rotation, scaling, and arbitrary affine transformations. For a class $\mathcal T$ of rationally parametrized transformations with $k$ degrees of freedom, we show that one can determine whether there is a transformation $\tau \in \mathcal T$ such that the Fréchet distance between $\pi$ and $\tau(\sigma)$ is at most $\delta$ in $\tilde{\mathcal{O}}(n^{3k+\frac{4}{3}})$ time.
\end{abstract}
    
\section{Introduction}

Many applications require determining the similarity of two geometric shapes, disregarding their location or orientation.
More specifically, \emph{shape matching} asks for the distance between two shapes if we allow the shapes to be transformed to minimize their distance~\cite{Alt09,AltG00,Veltkamp01}.
Transformations may include translations, rotations, scaling, or a combination thereof.
In this paper, we focus on polygonal curves under the Fr\'echet distance.
Curves occur in many applications and need to be matched whenever they describe a local pattern, for example to recognize handwritten characters~\cite{SriraghavendraKB07} or trademarks~\cite{AltSS07}. 
The Fréchet distance is arguably the most popular distance measure for curves in computational geometry.
There has been significant algorithmic progress on the Fr\'echet distance, for instance, most recently on approximation~\cite{ColombeF21,HorstO24,HorstKOS23}, data structures~\cite{AronovFKR24,BringmannDNP22,BuchinBG0W24,BuchinHOSSS22,filtser_static_2023,FiltserFK23,GudmundssonRSW21,GudmundssonSW23},
algorithm engineering~\cite{sigspatial1,sigspatial2,sigspatial3,DBLP:journals/jocg/BringmannKN21,DBLP:conf/esa/BringmannKN20},
simplification~\cite{BringmannC20,ChengH23,KerkhofKLMW19,KreveldLW20}, clustering~\cite{BruningCD22,BuchinDGHKLS19,BuchinDR23,GudmundssonW22,Nath020,DBLP:conf/gis/BuchinDLN19},  Fr\'echet variants~\cite{BuchinFLPRR23,BuchinFSW23,BuchinLOPUV23,DriemelHR22,fan_computing_2021,Filtser0MP23,FoxNPR24}, and its complexity in general~\cite{BlankD24,BuchinOS19,ChengH24,cheng_frechet_nodate}. 

The Fréchet distance under translation is defined as the minimum Fréchet distance for any translation of the curves.
Computing the Fr\'echet distance under translations was first studied by Efrat, Indyk and Venkatasubramanian~\cite{EfratIV04} and  by Alt, Knauer and Wenk~\cite{altMatchingPolygonalCurves2001,knauerComparingGeometricPatterns2002, wenkShapeMatchingHigher2003}. 
For two curves of total complexity $n$ in~$\mathbb R^2$, an $\tOh(n^{10})$ time algorithm\footnote{By $\tOh(\cdot)$ we hide (poly-)logarithmic factors in $n$.} for the Fr\'echet distance under translation was presented in~\cite{EfratIV04}, and an $\tOh(n^8)$ time algorithm was presented in~\cite{altMatchingPolygonalCurves2001}.
Wenk~\cite{wenkShapeMatchingHigher2003} generalized the approach in~\cite{altMatchingPolygonalCurves2001} to higher dimensions and a wide range of other transformations, e.g., for rotations or scalings in 2D in $\tOh(n^5)$ time, or for affine transformations in $d$ dimensions in $\tOh(n^{3(d^2+d)+2})$ time.

Despite significant algorithmic progress on various aspects of the Fr\'echet distance, no faster algorithms for computing the (continuous) Fr\'echet distance under transformations have been developed until now. In contrast, for computing the discrete Fr\'echet distance under translation, first an $\tOh(n^6)$-time algorithm~\cite{JiangXZ07} was developed, then an $\tOh(n^5)$-time algorithm~\cite{AvrahamKS15}, and more lately an $\tOh(n^{4+\frac{2}{3}})$-time algorithm~\cite{bringmannFrechetDistanceTranslation2021}. 

\paragraph*{Our Contribution}
In this paper, we present the first
progress on the Fr\'echet distance under transformations since its introduction~\cite{altMatchingPolygonalCurves2001,EfratIV04,knauerComparingGeometricPatterns2002,wenkShapeMatchingHigher2003}. Similar to the algorithm in~\cite{altMatchingPolygonalCurves2001,wenkShapeMatchingHigher2003}, our algorithm can be used for a wide range of classes of transformations. Specifically, given two curves $\pi$ and $\sigma$ of total complexity $n$ and a threshold $\delta \geq 0$, we want to determine whether there is a transformation $\tau$ from a given class of transformations, such that the Fr\'echet distance between $\pi$ and $\tau(\sigma)$ is at most $\delta$. Our algorithm improves the running time for translations in two dimensions from $\Oh(n^8)$ to $\tOh(n^{7 + \frac 1 3})$. More generally, we improve the running time for the various classes of transformations (and dimensions) given by Wenk~\cite{altMatchingPolygonalCurves2001} by roughly a factor of~$n^{2/3}$. For example, for rotations or scalings in $\mathbb R^2$, our algorithm runs in $\tOh(n^{4 + \frac 1 3})$ time and for affine transformations in $d$ dimensions in $\tOh(n^{3(d^2+d)+\frac 1 3})$ time.

We first present our approach for the special case of translations in two dimensions. Then we generalize our approach to other transformations and higher dimensions. Similarly to the approach  in~\cite{altMatchingPolygonalCurves2001}, we compute an arrangement in the space of transformations, which in the case of 2D translations has complexity $\Oh(n^6)$. In~\cite{altMatchingPolygonalCurves2001}, the Fr\'echet distance is computed for every vertex of this arrangement in $\Oh(n^2)$ time per vertex by using the \emph{free space diagram}, which results in an overall running time of $\Oh(n^8)$. In our approach, we instead traverse the arrangement, making use of the fact that the free space diagram for adjacent faces in the arrangement is similar.

The challenge with using this approach is that it requires a data structure for dynamic directed graph reachability.
This is a challenging problem on its own that to the best of our knowledge mostly saw progress on planar graphs~\cite{DBLP:conf/esa/DiksS07,karczmarz_et_al:LIPIcs.ICALP.2024.95,DBLP:conf/esa/Subramanian93} or more specifically grid graphs~\cite{AvrahamKS15,bringmannFrechetDistanceTranslation2021}.
In~\cite{bringmannFrechetDistanceTranslation2021}, a data structure with efficient updates and queries for reachability in a dynamic directed grid graph is presented. However, while the discrete Fr\'echet distance naturally reduces to a reachability problem on such a grid, this is not the case for the (continuous) Fr\'echet distance. An open question is whether there is a dynamic graph reachability data structure more suitable for the case of the continuous Fr\'echet distance under transformations.

We note that in an independent work, progress was made on the one-dimensional Fréchet distance under translation or scaling \cite{BCDKNR25}. In this work, the authors present an $\tOh(n^{8/3})$ algorithm for both problems, making use of the offline dynamic data structure of~\cite{bringmannFrechetDistanceTranslation2021}.

\paragraph*{Structure}
We provide an outline for the remainder of this paper. In Section~\ref{sec:prelim}, we cover the preliminaries for our paper. In Section~\ref{sec:fs-reach-to-fsg-reach}, we first present a detailed analysis of the changes in the free space diagram when traversing the arrangement in the space of translations. We refer to these changes as \emph{events}.
The free space diagram induces a directed graph, 
which we refer to as \emph{free space graph} and 
in which we need to perform a reachability query. We analyze how the graph changes for the various types of events. In Section~\ref{sec:fsg-to-gg}, 
we then show how these events can be modeled by a small number of changes to a suitable grid graph, which allows us to utilize the data structure in~\cite{bringmannFrechetDistanceTranslation2021}. The challenge is that at an event rows or columns in the free space graph may swap, or rows and columns may appear or disappear, which are operations that cannot be modeled directly in the grid graph.
In Section~\ref{sec:together}, we show how to generalize our algorithm to other classes of transformations.

\section{Preliminaries}\label{sec:prelim}

\newcommand{\width}[0]{\textbf{width}}
A $d$-dimensional polygonal curve $\pi$ is a piecewise linear curve represented as a continuous mapping from $[1, n]$ to $\reals^d$. For integer $i$, the point $\pi_i = \pi[i]$ is a vertex, and $\cseg{\pi_i} = \pi_i \pi_{i + 1} = \pi[i, i + 1]$ is a segment.

The Fréchet distance is a popular measure of the similarity between two polygonal curves. 
An \emph{orientation-preserving reparameterization} is a continuous and bijective function $f: [0, 1] \rightarrow [0, 1]$ such that $f(0) = 0$, and $f(1) = 1$. The $\width_{f, g}(\pi, \sigma)$ between two curves $\pi$ and $\sigma$ with respect to the reparameterizations $f$ and $g$, is defined as follows. 
\begin{align*}
    \width_{f, g}(\pi, \sigma) = \max_{t \in [0, 1]} \abse{\pi(f(t)) - \sigma(g(t))}
\end{align*}

\newcommand{\fredist}[1]{\delta_F(#1)}

Consider the scenario where a person is walking their dog with a leash connecting them: the person needs to stay on $\pi$ while walking according to $f$, and the dog needs to stay on $\sigma$ while walking according to $g$. The maximum leash length is the width between $\pi$ and $\sigma$ with respect to the reparameterizations $f$ and $g$. The standard \Frechet distance $\df{\pi, \sigma}$ is the minimum leash length required over all possible walks (defined by reparameterizations $f$ and~$g$).
\begin{align*}
    \df{\pi, \sigma} = \inf_{f, g \in [0, 1] \rightarrow [0, 1]} \width_{f, g}(\pi, \sigma)
\end{align*}

\newcommand{\freespaceDiagram}[3]{\mathcal{D}_{#3}(#1, #2)}
Problems relating to the \Frechet distance are commonly solved in a configuration space called the \emph{freespace diagram}.

The \emph{freespace} $\freespace{\pi, \sigma}{\delta}$ is a collection of points in $\reals^2$ in the range $[1, n] \times [1, n]$. A point $(x, y)$ is in the freespace if the Euclidean norm (distance) $\abse{\pi[x] - \sigma[y]}$ between $\pi[x]$ and $\sigma[y]$ is at most $\delta$.  As opposed to the free space, we will call $[1, n] \times [1, n] \setminus \freespace{\pi, \sigma}{\delta}$ the \emph{forbidden space}. 
\begin{align*}
    \freespace{\pi, \sigma}{\delta} = \{(x, y) \in [1, n] \times [1, n] \mid \abse{\pi(x) - \sigma(y)} \leq \delta\}
\end{align*}

The \emph{freespace diagram} $\D = \fd{\pi, \sigma}{\delta}$ is the partition of the freespace into $n \times n$ cells. A cell $\cell{i, j}$ in $\D$ contains the freespace  in the range $[i, i + 1] \times [j, j + 1]$. 
Alt and Godau~\cite{altComputingFrechetDistance1995} made the critical observation that the freespace within a cell $\cell{i, j}$ is an intersection between an ellipse $E_{i, j}$ and the square $S_{i, j} = [i, i + 1] \times [j, j + 1]$; it is convex with constant description complexity. They also showed that $\df{\pi, \sigma} \leq \delta$ if and only if there is a bi-monotone path in $\fd{\pi, \sigma}{\delta}$ from $(1, 1)$ to $(n, n)$ through the freespace. 

An intersection between the boundaries $\boundary{E_{i, j}}$ and $\boundary{S_{i, j}}\setminus\{(i, j), (i+1, j), (i, j+1), (i+1, j+1)\}$ is called a \emph{critical point}. A point $(i, j) \in \mathbb{N} \times \mathbb{N}$ is a \emph{corner point}. We say the line $\pi[i] \times [1, n]$ is the \emph{freespace diagram boundary} defined by $\pi_i$ (and analogous for $\sigma_i$). We say the strip $[i, i + 1] \times [1, n]$ is the $i$th \emph{column}, and the strip $[1, n] \times [i, i + 1]$ is the $i$th \emph{row}. We say a critical point $p$ is a row (resp. column) critical point if $p$ lies on a vertical (resp. horizontal) boundary.

\section{From freespace reachability to graph reachability} \label{sec:fs-reach-to-fsg-reach}

In this and the next section, 
we describe our ideas using an intuitive class of transformations: translation in $\reals^2$. Specifically, we determine whether there exists a translation vector $t = \lambda \Vec{v}$ such that $\df{\pi, \sigma + t} \leq \delta$, where $\lambda \in \reals_{\geq 0}$ and $\Vec{v}$ is a fixed directional vector. For this, we first describe how to transform the freespace reachability problem into a graph reachability problem. 

\sloppy Using the freespace diagram $\D = \fd{\pi, \sigma}{\delta}$, we construct a \emph{refined freespace diagram (refined FSD or FSD for short)} (see Figure~\ref{fig:fsd-to-fsg}, left). Let $l(\pi_i)$ (resp. $l(\sigma_j)$) be the vertical (resp. horizontal) boundary of $\D$ defined by $\pi_i$ (resp. $\sigma_j$).  For every critical point $p$ on the boundaries, we draw a perpendicular \emph{grid line} $l(p)$ through $p$. Note that by definition, a critical point does not coincide with a corner point of $\D$. If $p$ lies on the horizontal boundary, $l(p)$ is a vertical line; otherwise, $l(p)$ is a horizontal line. The refined FSD includes all grid lines, the freespace diagram boundaries, and their intersections. For simplicity, we redefine $\fd{\pi, \sigma}{\delta}$ to denote the refined FSD. Let the intersection between a grid line and a FSD boundary be a \emph{propagated critical point}. 

\begin{figure}[tbh]
    \centering
    \includegraphics[scale=0.75]{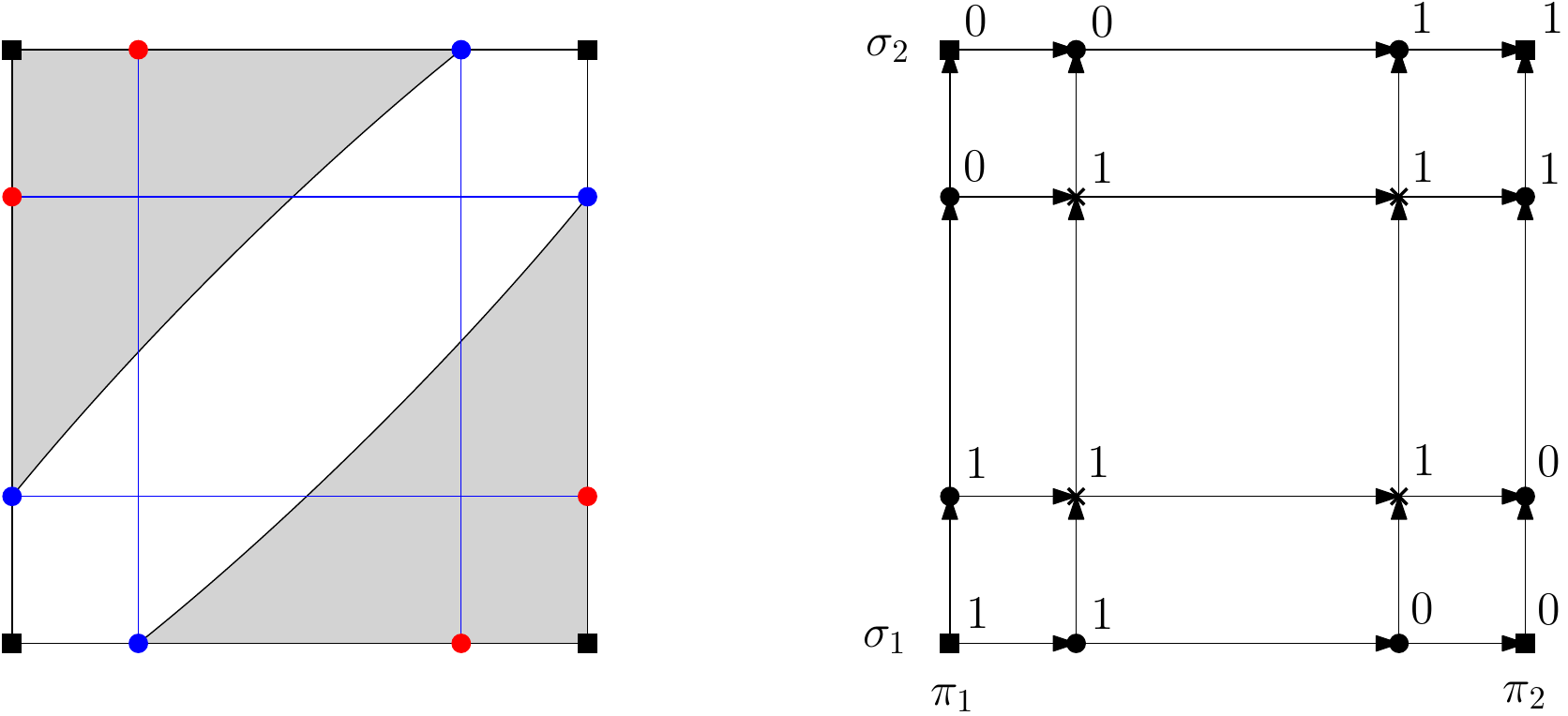}
    \caption{From the refined freespace diagram to the refined freespace graph. On the left freespace diagram (FSD), the corner points, critical points, and propagated critical points are marked by squares, blue circles, and red circles, respectively. On the right freespace diagram graph (FSG), the corner, boundary, and interior vertices are marked by squares, circles, and crosses, respectively.}
    \label{fig:fsd-to-fsg}
\end{figure}

Using the refined FSD $\fd{\pi, \sigma}{\delta}$, we construct a \emph{refined freespace graph (refined FSG or FSG for short)} $\fsgraph = \fsgraph_\delta(\pi, \sigma)$ as follows (see Figure~\ref{fig:fsd-to-fsg}, right). For every intersection between a grid line and a FSD boundary, add a \emph{boundary vertex}. For every intersection between two grid lines, add an \emph{interior vertex}. For every intersection between two freespace boundaries, add a \emph{corner vertex}. Note that each vertex in $\fsgraph$ is uniquely defined by an ordered pair $(p, q)$, where $l(p)$ is vertical and $l(q)$ is horizontal; $p$ (resp. $q$) is either a vertex of $\pi$ (resp. $\sigma$) or a critical point in $\D$ --- let $\vertex{p, q}$ be such vertex in $\fsgraph$. Each vertex is assigned a weight that is either 1 or 0. For every vertex $u = \vertex{p, q}$, if $a = l(p) \cap l(q)$ lies on a boundary of the freespace diagram, we set $\weight{u} = 1$ if $a$ lies in the freespace or $\weight{u} = 0$ if otherwise. We set the weights of the rest of the vertices, the interior vertices, to $1$. We say a vertex $u$ is \emph{activated} if $\weight{u} = 1$ or \emph{deactivated} if $\weight{u} = 0$. We say a vertex $\vertex{\pi_i, \cdot}$ (resp. $\vertex{\cdot, \sigma_j}$) lies on the freespace graph boundary defined by $\pi_i$ (resp. $\sigma_j$). 

To construct edges, we use the geometric positions of the grid lines and the FSG boundaries. For vertices in $\fsgraph$ defined by every grid line or FSG boundary $l(p)$, add a directed edge $(\vertex{p, q}, \vertex{p, q'})$ to $\fsgraph$ if $l(q)$ is immediately below $l(q')$. If $l(q)$ and $l(q')$ overlap, break ties arbitrarily. Analogously, add a directed edge $(\vertex{p, q}, \vertex{p', q})$ if $l(p)$ is immediately to the left of $l(p')$. A path $P \subseteq \fsgraph$ is an ordered subset $\{(a_1, b_1), ..., (a_m, b_m)\}$ of edges where $b_i = a_{i + 1}$, $\forall 1 \leq i < m$. We say $P$ is a \emph{feasible path} if $a_1 = \vertex{\pi_1, \sigma_1}$, $b_m = \vertex{\pi_n, \sigma_n}$, and $\weight{a_i} = \weight{b_i} = 1$ for all values of $i$. When $a_1$ and $b_m$ are specified, $P$ is a feasible path if $P$ uses exclusively activated vertices. We say a FSG $\fsgraph$ is \emph{$st$-reachable} if there exists a feasible path in $\fsgraph$ from $\vertex{\pi_1, \sigma_1}$ to $\vertex{\sigma_n, \pi_n}$. The following lemma naturally is derived from the properties of the freespace diagram. 
\begin{restatable}{lemma}{fsgraphpath} \label{lem:fsgraph-path}
    The FSG $\fsgraph = \fsgraph_\delta(\pi, \sigma)$ is $st$-reachable if and only if $\df{\pi, \sigma} \leq \delta$. 
\end{restatable}
\begin{proof}
    It is well-known that $\df{\pi, \sigma} \leq \delta$ if and only if there exists an $xy$-monotone path through the freespace from $(1, 1)$ to $(n, n)$ in $\fd{\pi, \sigma}{\delta}$. Such a monotone path can always be transformed into a feasible path $P$. Every subpath $P_{i, j} = P \cap \cell{i, j}$ within cell $\cell{i, j}$ is a line segment $ab$, where both $a$ and $b$ is either a critical point or a corner point. We will argue that for every subpath $P_{i, j}$ from critical point $a = l(p) \cap l(q)$ to $b = l(p') \cap l(q')$, there exists a path $P^{f}_{i, j} \subseteq \fsgraph$ from $\vertex{p, q}$ to $\vertex{p', q'}$ using only vertices of weight $1$. By construction, $\weight{\vertex{p, q}} = \weight{\vertex{p', q'}} = 1$ and all interior vertices have weight $1$. Therefore, a path starting at $\vertex{p, q}$, turning at $\vertex{p', q}$, and finally arriving at $\vertex{p', q'}$ is a valid candidate for $P^{f}_{i, j}$, as required. 

    Analogously, if a subpath $P^{f}_{i, j}$ from $\vertex{p, q}$ to $\vertex{p', q'}$ exists, then we use the segment $ab$ as $P_{i, j}$. The segment $ab$ must lie in the freespace, because the freespace is convex within every cell $\cell{i, j}$, as shown by Alt and Godau~\cite{altComputingFrechetDistance1995}. 
\end{proof}

Consider translating $\sigma$ along $\Vec{v}$. As $\sigma$ translates, the FSG also changes, and we would like to know how these changes affect the $st$-reachability of $\fsgraph$. To do this, we track the following freespace events. See Figure~\ref{fig:fs-events} for visualizations of these row freespace events. 

\begin{definition} \label{def:fs-events}
    We define the following \emph{row freespace events}. Column freespace events are defined analogously. 
    \begin{enumerate}
        \item A \emph{Vertex-edge (VE) event} is defined by a tuple $(p, \cseg{\sigma_w})$, where $p$ is either a vertex of $\pi$ or a row critical point in row~$w$. 
        \begin{enumerate}
            \item Entering/leaving event: the corner point $p$ enters or leaves the freespace. 
            \item Appearing/disappearing event: the grid line $l(p)$ appears or disappears as the critical point $p$ appears or disappears. 
        \end{enumerate}
        \item A \emph{Vertex-vertex-edge (VVE) event} is defined by a triplet $(p, q, \cseg{\sigma_w})$, where $p$ and $q$ are row critical points in row~$w$. 
        \begin{enumerate}
            \item Overlapping event: two grid lines $l(p)$ and $l(q)$ overlap. 
            \item Separating event: two overlapping grid lines $l(p)$ and $l(q)$ no longer overlap. 
        \end{enumerate}
    \end{enumerate}

    \begin{figure}[tbh]
        \centering
        \includegraphics[width=\textwidth]{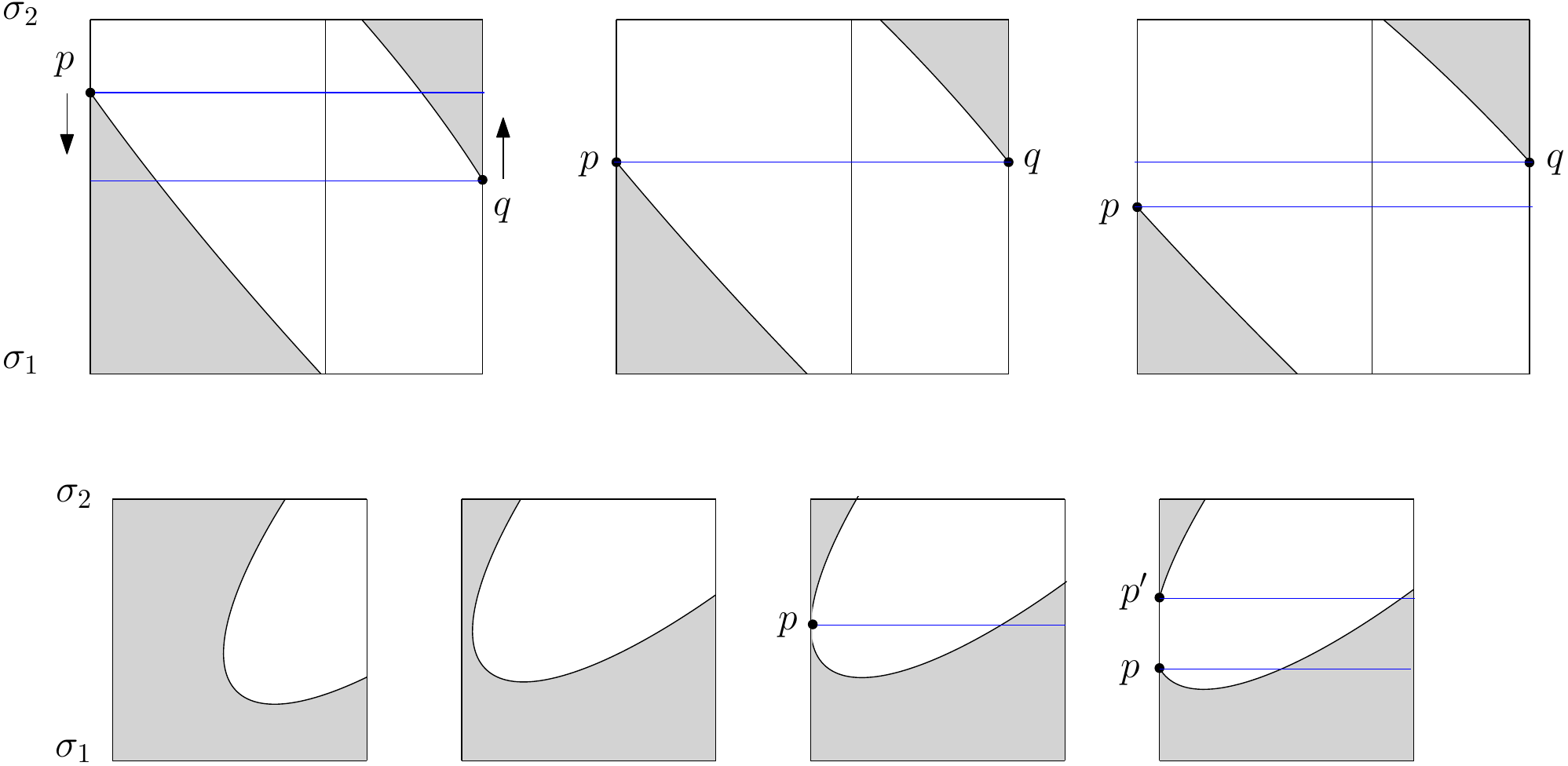}
        \caption{In the top row of figures, as the segment $\sigma_1 \sigma_2$ translates, the critical points $p$ and $q$ moves down and up, respectively. As $l(p)$ and $l(q)$ (colored in blue) move, they overlap and then separate. In the bottom row, as $\sigma_1 \sigma_2$ translates, a new critical point $p$ appears, and then a new critical point $p'$ appears. }
        \label{fig:fs-events}
    \end{figure}
\end{definition}

Let $\sigma + t$ be the curve by applying translation $t$ to $\sigma$. For now, we assume that we know the ordered set $T = \{t_1, t_2, ..., t_m\}$ of freespace events (translations), where $t_i = \lambda_i \Vec{v}$, $\lambda_i \in \reals_{\geq 0}$, and $\lambda_i \leq \lambda_{i + 1}$. We further assume that $T$ is \emph{complete}, which is defined as follows. 
\begin{definition} \label{def:complete-event}
    We say the ordered set $T = \{t_1, ..., t_m\}$ of freespace events is \emph{complete} if $\fsgraph_\delta(\pi, \sigma + t_i)$ and $\fsgraph_\delta(\pi, \sigma + t_{i + 1})$ differ by exactly one freespace event, for all $1 \leq i < m$. 
\end{definition}

To update the freespace graph to reflect the state of the freespace diagram, we define the following freespace graph operations. Let $\row{p}$ be the set of vertices defined by $l(p)$, and let $\row{p}[i]$ denote the $i$th vertex, where $i \geq 1$. 
\begin{definition} \label{def:fsg-operations}
    We define the following \emph{freespace graph operations}. 
    \begin{itemize}
        \item Vertex operation: either activate or deactivate a vertex of $\fsgraph$. 
        \item Row operation: either insert or delete a row of $\fsgraph$. Column operations are defined analogously. 
        \begin{itemize}
            \item To insert a row $\row{p}$ of vertices between adjacent rows $\row{p_a}$ and $\row{p_b}$, add a vertex $\vertex{q, p}$ for every $q$, where $q$ is either a critical point on a horizontal FSG boundary or a vertex of $\pi$. For every $1 \leq i < \abs{\row{p_a}}$, 
            \begin{enumerate}
                \item remove the edge $(\row{p_b}[i], \row{p_a}[i])$, 
                \item add horizontal edge $(\row{p}[i], \row{p}[i + 1])$, and
                \item add vertical edges $(\row{p_b}[i], \row{p}[i])$ and $(\row{p}[i], \row{p_a}[i])$. 
            \end{enumerate}
            \item To remove the row $\row{p}$ of vertices, remove $\row{p}$ and their adjacent edges. For every $1 \leq i \leq \abs{\row{p_a}}$, add edges $(\row{p_b}[i], \row{p_a}[i])$.
        \end{itemize}
    \end{itemize}
\end{definition}

We show that we can update the FSG $\fsgraph_\delta(\pi, \sigma + t_{i - 1})$ to $\fsgraph_\delta(\pi, \sigma + t_{i})$ using a constant number of freespace graph operations, plus processing time. For these updates, we will distinguish between a corner vertex operation and a boundary vertex operation.  

\begin{lemma} \label{lem:fsgraph-event-updates}
    Let $\fsgraph_i = \fsgraph_\delta(\pi, \sigma + t_i)$. Given $\fsgraph_{i - 1}$, to compute $\fsgraph_i$, it takes 
    \begin{itemize}
        \item $\Oh(1)$ time and $\Oh(1)$ corner vertex operations if $t_{i}$ is an entering/leaving event, 
        \item $\Oh(1)$ time and $\Oh(1)$ boundary vertex operations if $t_i$ is an overlapping or separating event, 
        \item $\Oh(n)$ time plus $\Oh(1)$ row operations, if $t_i$ is an appearing/disappearing event. 
    \end{itemize}
\end{lemma}

\begin{proof}
    If $t_i$ is an entering/leaving event, we simply set the weight of the respective corner vertex to either 1 or 0. 

    If $t_i$ is a VVE event defined by $(p, q, \cseg{w})$, let $p$ and $q$ lie on the $i$th and $j$th FSD boundary, respectively. A VVE event affects only the vertices $\vertex{\pi_i, p}$, $\vertex{\pi_j, p}$, $\vertex{\pi_i, q}$, and $\vertex{\pi_j, q}$. It takes $O(1)$ time to determine their weights by computing cells $\cell{i, w}$ and $\cell{j, w}$ from scratch. 

    If $t_i$ is an appearing/disappearing event in the $j$th row, it takes linear time to recompute the row critical points in this row. Then, it takes linear time to compute the grid lines $l(p_a)$ and $l(p_b)$ that $l(p)$ lie between. Once $l(p_a)$ and $l(p_b)$ is identified, it takes exactly one row insertions or deletions to update $\fsgraph_{i - 1}$ to $\fsgraph_i$. Therefore, it takes $\Oh(n)$ time plus a constant number of row operations.
\end{proof}

Using a naive implementation, in the worst case, a row operation can take $\Omega(n^2)$ time since there are $\Omega(n^2)$ vertical grid lines. Determining $st$-reachability takes $\Omega(n^4)$ time, since there are~$\Omega(n^4)$ vertices in the FSG.
Hence, using the FSG directly would be infeasible.
In the next section, we fix these issues by defining an \doublequote{equivalent} grid graph in which updates and queries can be done more efficiently than the naive implementation. 

\section{From freespace graph reachability to grid graph reachability} \label{sec:fsg-to-gg}
In this section, using the refined FSG $\fsgraph = \fsgraph_\delta(\pi, \sigma)$, we define a \emph{grid graph} $\ggraph = \ggraph_\delta(\pi, \sigma)$. We then show that $\fsgraph$ is $st$-reachable if and only if $\ggraph$ is $st$-reachable. An advantage of the newly defined grid graph is that its number of vertices does not change under updates. We show that the structural changes implied by the FSG operations can be simulated by simply modifying the weights in the grid graph, without needing to add or remove vertices. These features of the grid graph allow us to use the offline dynamic grid reachability result in~\cite{bringmannFrechetDistanceTranslation2021} to obtain a faster update and query time.

The grid graph $\ggraph$ contains all vertices of the the freespace graph $\fsgraph$. Additionally, to maintain the same number of vertices in the grid graph under updates, we require a set of placeholder vertices, which we define as follows. Refer to Figure~\ref{fig:grid-graph} for an illustration. Let $m^r_i$ (resp. $m^c_i)$ be the number of row critical points in the $i$th row (resp. column) of $\D = \fd{\pi, \sigma}{\delta}$. We define a family of ordered sets $\{H^r_1, ..., H^r_{n - 1}\}$ of row placeholder points. Each set $H^r_i$ contains exactly $2n - m^r_i$ points, and let $H^r_i[j]$ denote the $j$th point for $1 \leq j \leq 2n - m^r_i$. The family of ordered sets of column placeholder points $\{H^c_1, ..., H^c_{n - 1}\}$ are analogously defined. For a row (resp. column) placeholder point $h$, the \emph{placeholder line} $l(h)$ is horizontal (resp. vertical). Note that the placeholder points and lines are abstractly defined as they do not exist in $\D$. 

\begin{figure}[tbh]
    \centering
    \includegraphics[scale=0.75]{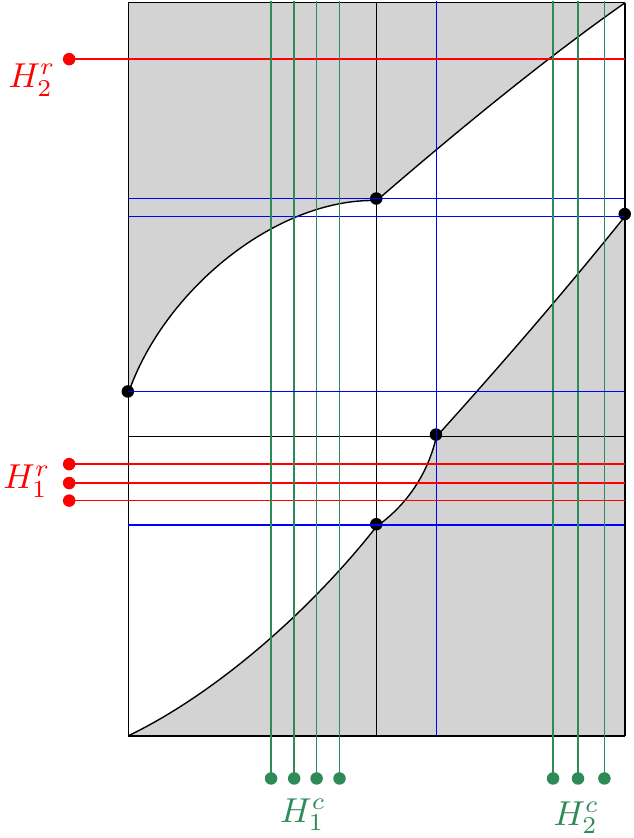}
    \caption{A visual illustration of the placeholder lines and their relative positions among the grid lines and the freespace diagram boundaries. The row (resp. column) placeholder points and lines are colored in red (resp. green). The grid lines are colored in blue.}
    \label{fig:grid-graph}
\end{figure}

In addition to vertices in $\fsgraph$, add a \emph{placeholder vertex} $\vertex{p, q}$ to $\ggraph$ where $p$ (resp. $q$) is either a vertex of $\pi$ (resp. $\sigma$), a critical point, or a placeholder point. The weight of a placeholder vertex on a boundary matches the weight of the adjacent corner vertex either directly above or to its right. Specifically, for every $h$ in every $H^r_i$ and every point $p$ in the union of the vertices of $\pi$ and row critical points. For $h \in H^r_j$, we set $\weight{\vertex{p, h}} = \weight{\vertex{\pi_i, \sigma_{j + 1}}}$ if $p = \pi_i$. Otherwise, we set $\weight{\vertex{p, h}} = 1$. Analogously, for every $h \in H^c_i$, we set $\weight{\vertex{h, q}} = \weight{\vertex{\pi_{i + 1}, \sigma_j}}$ if $q = \sigma_j$. Otherwise, we set $\weight{\vertex{h, q}} = 1$. 

To define the edges in $\ggraph$, we define the ordering of the placeholder lines among the grid lines and the FSD boundaries. The lowest placeholder line $l(H^r_i[1])$ lies above all grid lines $l(p)$, where $p$ is a row critical point on the $i$th row. The placeholder line $l(H^r_i[j + 1])$ is above $l(H^r_i[j])$. The FSD boundary $l(\sigma_{i + 1})$ is above the highest placeholder line $l(H^r_i[2n - m_i^r])$. The ordering involving vertical placeholder lines is analogously defined. The set of placeholder lines defined by $H^c_i$ are between $l(\pi_{i + 1})$ and the rightmost grid line in the $i$th column. 

For vertices in $\ggraph$ defined by $l(p)$, add a directed edge $(\vertex{p, q}, \vertex{p, q'})$ to $\ggraph$ if $l(q)$ is immediately below $l(q')$ and add a directed edge $(\vertex{q, p}, \vertex{q', p})$ to $\ggraph$ if $l(q)$ is immediately to the left of $l(q')$. If $l(q)$ and $l(q')$ overlap, then break ties using the same ordering as in $\fsgraph$. Add additional diagonal directed edge $(\vertex{p, q}, \vertex{p', q'})$ if $p$ is immediately to the left of $p'$ and $q$ is immediately below $q'$. We say a vertex defined by two placeholder points is a placeholder vertex as well as an interior vertex. 

Next, we check that our construction fits the grid graph definition of~\cite{bringmannFrechetDistanceTranslation2021}. An $N \times N$ grid graph consists of vertices numbered from $(1, 1)$ to $(N, N)$ and edges from vertex $(i, j)$ to each of vertices $(i + 1, j)$, $(i, j + 1)$, and $(i + 1, j + 1)$, where the weight of a vertex is either $1$ or $0$. With $\ggraph_{\delta}(\pi, \sigma)$ defined, we prove the following. By construction, every of $n - 1$ rows in $\fd{\pi, \sigma}{\delta}$ contains exactly $2n$ critical points and placeholder points combined. Together with $n$ horizontal boundaries, there are $N = 2n \cdot (n - 1) + n = 2n^2 - n$ horizontal lines. For the same reasons, there are $N$ vertical lines. Clearly, $\ggraph_{\delta}(\pi, \sigma)$ is an $N \times N$ grid graph. 

Before proving that $\fsgraph_\delta(\pi, \sigma)$ and $\ggraph_\delta(\pi, \sigma)$ are equivalent in terms of $st$-reachability, we first show that a feasible path $P \subseteq \ggraph$ has several desired properties. By the monotonicity of a feasible path and the convexity of the freespace in a cell, we observe the following. 
\begin{restatable}{observation}{corneradjacentboundary} \label{obs:corner-adjacent-to-boundary-1}
    Let $P$ be a feasible path in $\ggraph$. If $P$ contains a corner vertex $\vertex{\pi_i, \sigma_j}$, then let $l(p_l)$ and $l(p_r)$ be the first grid lines that are not placeholder lines to the left and right of $l(\pi_i)$, respectively. Similarly, let $l(p_a)$ and $l(p_b)$ be the first grid line above and below $l(\sigma_j)$ respectively. Then, we have either $\weight{\vertex{p_l, \sigma_j}} = 1$ or $\weight{\vertex{\pi_i, p_b}} = 1$. Analogously, we have either $\weight{\vertex{\pi_i, p_a}} = 1$ or $\weight{\vertex{p_r, \sigma_j}} = 1$. 
\end{restatable}
\begin{proof}
    \begin{figure}[tbh]
        \centering
        \includegraphics[scale=0.75]{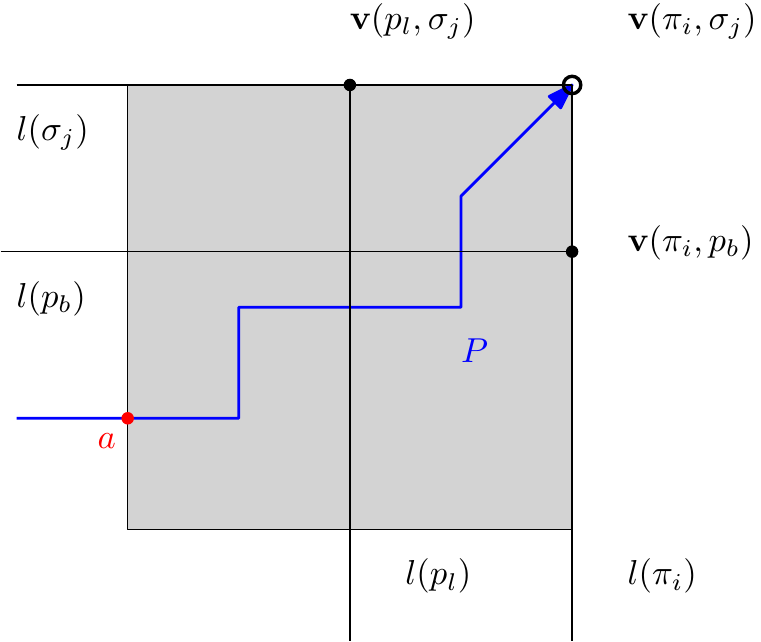}
        \caption{If vertices $\vertex{\pi_i, p_b}$ and $\vertex{p_l, \sigma_j}$ are deactivated, point $(i, j)$ is the only point that lies in the freespace in $\cell{i - 1, j - 1}$. Due to monotonicity, the path $P$ must use a deactivated vertex (say $a$).}
        \label{fig:free-adjacent-to-corner}
    \end{figure}

    See Figure~\ref{fig:free-adjacent-to-corner} for an illustration. For the sake of contradiction, assume that the weight of both $\vertex{p_l, \sigma_j}$ and $\vertex{\pi_i, p_b}$ are $0$. Since the freespace is convex within a cell, $l(\pi_i) \cap l(\sigma_j)$ is the only point in $\cell{i - 1, j - 1}$ that lies in the freespace. By construction, this implies that for all grid lines $l(p)$ lying between $l(\pi_{i - 1})$ and $l(\pi_{i})$, $\weight{\vertex{p, \sigma_{j - 1}}} = 0$. For all $l(q)$ lying between $l(\sigma_{j -  1})$ and $l(\sigma_{j})$, $\weight{\vertex{\pi_{i - 1}, q}} = 0$. There exists at least one vertex with weight $0$ in $P$, contradicting the assumption that $P$ uses exclusively vertices of weight $1$. An analogous argument holds for the case where the weights of both $\vertex{p_r, \sigma_j}$ and $\vertex{\pi_i, p_a}$ are $0$. 
\end{proof}

Due to the properties of $\ggraph$, a diagonal edge in a path $P$ can be replaced by using the rectilinear edges in the same \doublequote{cube}. Furthermore, if the final vertex $b$ of a subpath $P_j \subseteq P$ is a placeholder vertex, we can transform $P_j$ such that it ends at a non-placeholder vertex. We have the following lemma. 

\begin{restatable}{lemma}{feasiblepathproperties} \label{lem:feasible-path-properties}
    If there is a feasible path $P$ in a grid graph $\ggraph$, there is a feasible path $P'$ in $\ggraph$ with the following properties. 
    \begin{enumerate}
        \item $P'$ does not contain diagonal edges. 
        \item For every $1 \leq j \leq n$, the first vertex of $P'$ partly defined by $\sigma_j$ is not a placeholder vertex.
    \end{enumerate}    
\end{restatable}
\begin{proof}
    We first replace the diagonal edges in $P$. Consider a diagonal edge $(a, b) \in P$ in the \doublequote{cube}, where $(a, c_t)$ and $(c_b, b)$ are the vertical edges, and $(a, c_b)$ and $(c_t, b)$ are horizontal edges. Observe that by construction, either both $a$ and $b$ lie on the boundaries, or at least one of them is an interior vertex. If both $a$ and $b$ lie on the boundaries, either $c_b$ or $c_t$ must be activated, since the opposite would contradict either Observation~\ref{obs:corner-adjacent-to-boundary-1} or the convexity of the freespace in a cell. 
    
    If $a$ is an interior vertex, we consider the following cases. If $b$ is also an interior vertex, then so are $c_t$ and $c_b$ with weight 1, and we replace $(a, b)$ with $(a, c_t)$ and $(c_t, b)$. If $b$ is a boundary vertex, then either $c_b$ or $c_t$ is an interior vertex, and we replace $(a, b)$ with either $(a, c_b)$ and $(c_b, b)$ or $(a, c_t)$ and $(c_t, b)$. The more interesting case is when $b = \vertex{\pi_{i + 1}, \sigma_{j + 1}}$ is a corner vertex. In this case, $c_b = \vertex{\pi_{i + 1}, h \in H^r_i}$, and by construction, $\weight{c_b} = \weight{b} = 1$. We replace $(a, b)$ with $(a, c_b)$ and $(c_b, b)$. 

    If $b$ is an interior vertex, we use similar case distinctions. If $a$ is an interior vertex, $c_b$ must also be an interior vertex. If $a$ is a boundary vertex lying on the left (resp. bottom) boundary, then $c_b$ (resp. $c_t$) is an interior vertex. The more interesting case is where $a = \vertex{\pi_i, \sigma_j}$ is a corner vertex. In this case, both $c_b$ and $c_t$ are boundary vertices. By Observation~\ref{obs:corner-adjacent-to-boundary-1}, at least one of them (say $c_b$) has weight $1$, and we replace $(a, b)$ with $(a, c_b)$ and $(c_b, b)$. 

    \begin{figure}[tbh]
        \centering
        \includegraphics[scale=0.75]{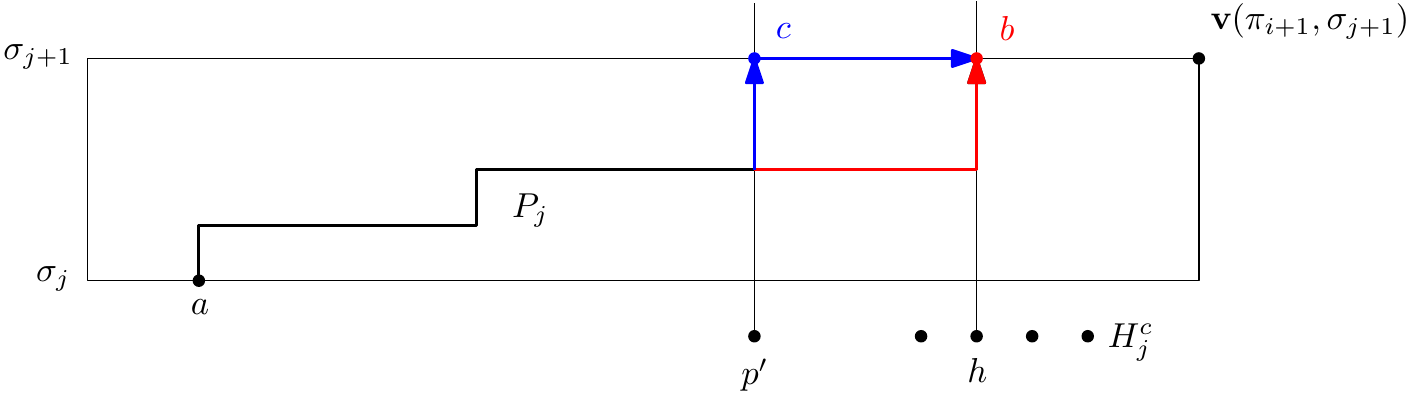}
        \caption{A path $P_{j - 1}$ ending at a placeholder vertex $b$ can be transformed to end at a non-placeholder vertex $c$.}
        \label{fig:path-starts-non-placeholder}
    \end{figure}

    With $P$ containing exclusively rectilinear edges, we partition $P$ into subpaths lying on different rows (see Figure~\ref{fig:path-starts-non-placeholder}). Specifically, let $P_j \subseteq P$ be the subpath containing the first vertex $a$ defined by $\sigma_j$ and the first vertex $b$ defined by $\sigma_{j + 1}$. We first transform $P_j$ to guarantee that $b$ is not defined by a placeholder point. Let $b = \vertex{h \in H^c_i, \sigma_{j + 1}}$, and note that $\weight{b} = \weight{\pi_{i + 1}, \sigma_{j + 1}} = 1$. Let $l(p')$ be immediately to the left of $l(H^c_i[1])$. By Observation~\ref{obs:corner-adjacent-to-boundary-1}, $\weight{c = \vertex{p', \sigma_{j + 1}}} = 1$. Combining this argument with the fact that all interior vertices have weight $1$ and the rectilinearity of $P$, we can transform $P_j$ to end at $c$, and $P_{j + 1}$ to start at $c$. Since the first vertex $\vertex{\pi_1, \sigma_1}$ and the final vertex $\vertex{\pi_n, \sigma_n}$ of $P$ are not defined by a placeholder vertex, once we apply the transformation above, the first vertex of every subpath $P_j$ is not defined by a placeholder vertex. The proof is complete. 
\end{proof}

We next observe that if there is a path in $\ggraph$ that does not use a placeholder vertex, the path also exists in $\fsgraph$. Indeed, excluding the placeholder lines, $\ggraph$ and $\fsgraph$ use the same set of grid lines and FSG boundaries, and the same ordering. 
\begin{observation} \label{obs:gg-fs-path-in-a-cell}
    If there is a path $P$ in $\ggraph = \ggraph_\delta(\pi, \sigma)$ such that $P$ does not use any placeholder vertices or diagonal edges, then $P$ also exists in $\fsgraph = \fsgraph_\delta(\pi, \sigma)$. 
\end{observation}

To demonstrate that $\fsgraph$ and $\ggraph$ are equivalent with respect to $st$-reachability, we first note that any feasible path in $\ggraph$ corresponds to a feasible path in $\fsgraph$. Specifically, given a subpath $P \subseteq \ggraph$ that traverses the \doublequote{strip} defined by a single set of placeholder points, we can always construct a corresponding subpath $Q \subseteq \fsgraph$ such that $Q$ starts and ends at the same vertices as $P$. We have Lemma~\ref{lem:ggraph-column-cross-path} and~\ref{lem:ggraph-placeholder-row-path}. 

\begin{restatable}{lemma}{ggraphcolumncrosspath} \label{lem:ggraph-column-cross-path}
    Let $P$ be a path in $\ggraph = \ggraph_{\delta}(\pi, \sigma)$. Let $P$ start at a non-placeholder vertex $\vertex{p, \sigma_j}$. Let $(\vertex{p', q'}, \vertex{p', H^r_j[1]})$ be the last edge of $P$, where $q' \neq H^r_j[1]$. There exists a path $Q$ from $\vertex{p, \sigma_j}$ to $\vertex{p', q'}$ in $\fsgraph = \fsgraph_{\delta}(\pi, \sigma)$.  
\end{restatable}
\begin{proof}
    \begin{figure}[tbh]
        \centering
        \includegraphics[width=0.8\textwidth]{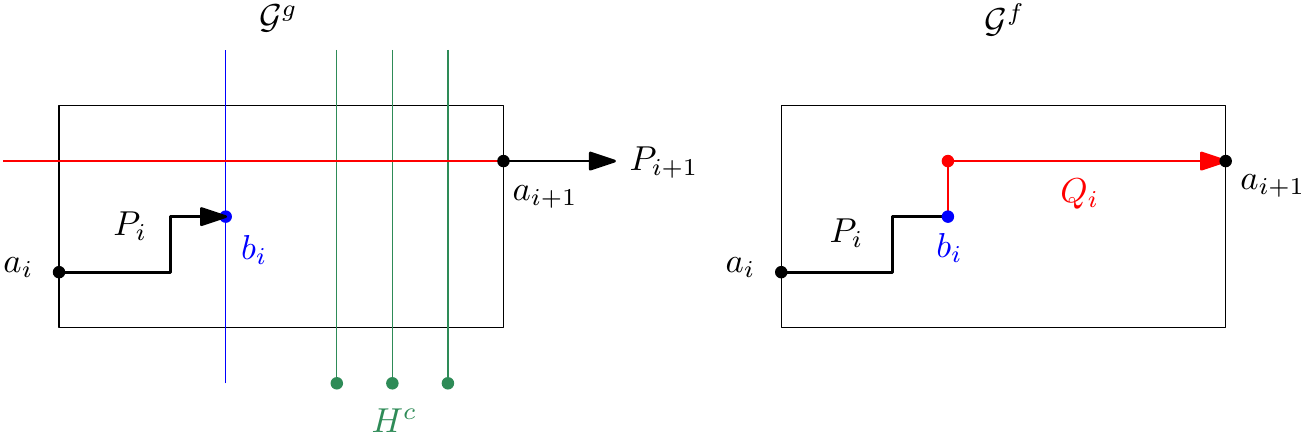}
        \caption{A path $Q_i$ in $\fsgraph$ can be constructed by connecting the vertex $b_i$ where $P_i$ ends and the vertex $a_{i + 1}$ where $P_{i + 1}$ starts.}
        \label{fig:path-cross-vert-placeholder-line}
    \end{figure}

    Let $\{P_1, ..., P_u\}$ be the subpaths of $P$ generated by removing all placeholder vertices from $P$ (see Figure~\ref{fig:path-cross-vert-placeholder-line}). For $1 \leq i \leq u$, let $P_i$ be a path starting from $a_i$ and ending at $b_i$. By Observation~\ref{obs:gg-fs-path-in-a-cell}, the path $P_i$ also exists in $\fsgraph$. Since the placeholder vertices are removed, every $a_i$ or $b_i$ is either a boundary vertex or an interior vertex. Since the interior vertices have weight $1$, a path $Q_i$ from $b_i$ to $a_{i + 1}$ exists in $\fsgraph$. We set  $Q = (\bigcup_{1 \leq i \leq u - 1} P_i \cup Q_i) \cup P_u$ to complete the proof. 
\end{proof}

\begin{restatable}{lemma}{ggraphrowcrosspath} \label{lem:ggraph-placeholder-row-path}
    Let $P$ be a path in $\ggraph = \ggraph_{\delta}(\pi, \sigma)$. Let $(\vertex{p, q}, \vertex{p, H^r_j[1]})$ be the first edge of $P$, and let $P$ ends at a non-placeholder vertex $\vertex{p', \sigma_{j + 1}}$. There exists a path $Q$ from $\vertex{p, q}$ to $\vertex{p', \sigma_{j + 1}}$ in $\fsgraph = \fsgraph_{\delta}(\pi, \sigma)$.  
\end{restatable}
\begin{proof}
    If $P$ does not contain a vertex defined by any vertical boundary, then the lemma trivially holds as the interior vertices have weight $1$. Otherwise, $P$ contains a set of vertices defined by $\{\pi_u, ..., \pi_w\}$ in increasing order of indices. By Observation~\ref{obs:corner-adjacent-to-boundary-1} and convexity of the freespace in a cell, there is a path $P_1$ in $\fsgraph$ from $\vertex{p, q}$ to $\vertex{\pi_u, \sigma_{j + 1}}$, and a path $P_2$ from $\vertex{\pi_w, \sigma_{j + 1}}$ to $\vertex{p', \sigma_{j + 1}}$. 
    
    Since $P$ uses only activated vertices, for every $u \leq i \leq w$, $\weight{\vertex{\pi_i, h}} = 1$ for some $h \in H^r_j$. By construction of $\ggraph$, $\vertex{\pi_i, h}$ is activated if and only if $\vertex{\pi_i, \sigma_{j + 1}}$ is activated, whose weight must also be $1$. By the convexity of the freespace within a cell, since $\vertex{\pi_i, \sigma_{j + 1}}$ and $\vertex{\pi_{i + 1}, \sigma_{j + 1}}$ are activated, for every $l(p)$ lying between $l(\pi_i)$ and $l(\pi_{i + 1})$, the vertex $\vertex{p, \sigma_{j + 1}}$ is activated. Therefore, $\forall u \leq i \leq w$, there is a path $P_i \subset \fsgraph$ from $\vertex{\pi_i, \sigma_{j + 1}}$ to $\vertex{\pi_{i + 1}, \sigma_{j + 1}}$ using activated vertices. We set $Q = P_1 \cup (\bigcup_{u \leq i \leq w} P_i) \cup P_2$ to complete the proof. 
\end{proof}

We are finally ready to show that $\fsgraph$ and $\ggraph$ are equivalent in terms of $st$-reachability. 
\begin{lemma} \label{lem:fs-gg-path-equivalent}
    There exists a feasible path $P^f$ in the freespace graph $\fsgraph_\delta(\pi, \sigma)$ if and only if there exists a feasible path $P^g$ in the grid graph $\ggraph_\delta(\pi, \sigma)$. 
\end{lemma}

\begin{proof}
    First, we observe that if there is a feasible path $P^f \subseteq \fsgraph$, then there is a feasible path $P^g \subseteq \ggraph$. 
    Consider a partition of $P^f$ into subpaths lying in different cells of the freespace diagram. Let the subpath $P^f_{i, j} \subseteq P$ start from the vertex $a$ to the vertex $b$, where $a$ lies on the bottom or left boundary of $\cell{i, j}$, and $b$ lies on the top or right boundary of $\cell{i, j}$. Since the interior vertices have weight 1, and a diagonal edge can be used if $b$ is a corner vertex, there is a path $P^g_{i, j} \subseteq \ggraph$ from $a$ to $b$ using activated vertices. 

    Second, we show that if there is a feasible path $P^g \subseteq \ggraph$, then there is a feasible path $P^f \subseteq \fsgraph$. We use an analogous argument where we construct a feasible path $P^f$ using subpaths in $P^g$. By Lemma~\ref{lem:feasible-path-properties}, we know that $P^g$ uses exclusively rectilinear edges. Furthermore, by the same Lemma~\ref{lem:feasible-path-properties}, $P^g$ can be partitioned into subpaths $\{P^g_1, ..., P^g_{n - 1}\}$ such that $\forall 1 \leq j \leq n - 1$, $P^g_j$ starts at a non-placeholder vertex $a = \vertex{p, \sigma_j}$, and ends at another non-placeholder vertex $b = \vertex{p', \sigma_{j + 1}}$. We partition $P^g_j$ using its first edge $(a', b')$, where vertex $a'$ is partly defined by a placeholder point $H^r_j[1]$. By Lemma~\ref{lem:ggraph-column-cross-path}, there is a path $P_j^{fa}$ in $\fsgraph$ from $a$ to $a'$. By Lemma~\ref{lem:ggraph-placeholder-row-path}, there is a path $P_j^{fb}$ in $\fsgraph$ from $a'$ to $b$. We set $P^f_j = P^{fa}_j \cup P^{fb}_j$ and $P^f = \bigcup_{j} P^f_j$ to complete the proof.  
\end{proof}

Now, we show that freespace graph operations can be implemented in the grid graph efficiently. 
\begin{lemma} \label{lem:ggraph-number-of-update-per-op}
    Let $\ggraph_i$ be the associated grid graph of the freespace graph $\fsgraph_i$. Let $u$ be a freespace graph operation that updates $\fsgraph_1$ to $\fsgraph_2$. To update $\ggraph_1$ to $\ggraph_2$, it is sufficient to update the weights of at most: 
    \begin{itemize}
        \item $\Oh(n)$ vertices if $u$ is a corner vertex operation,
        \item $\Oh(1)$ vertices if $u$ is a boundary vertex operation, or
        \item $\Oh(n)$ vertices if $u$ is a row operation. 
    \end{itemize}
\end{lemma}

\begin{proof}
    If $u$ is a corner vertex operation activating a corner vertex $\vertex{\pi_i, \sigma_j}$, we activate $\vertex{\pi_i, \sigma_j}$ in $\ggraph$. Then, we activate $\vertex{\pi_i, h_1}$ for every $h_1 \in H^r_{j - 1}$, and activate $\vertex{h_2, \sigma_j}$ for every $h_2 \in H^c_{i - 1}$. Since there are at most a linear number of placeholder points defined per row and column, this operation requires us to change the weights of $\Oh(n)$ vertices in $\ggraph$. If $u$ activates a boundary vertex $a$, we simply activate $a$ in $\ggraph$. If $u$ deactivates a vertex, we use analogous procedure. 
    
    If $u$ is a boundary vertex operation, to insert a row of vertices while maintain the properties of the grid graph, we take advantage of the fact that the intersection between the freespace and a cell boundary is a single interval. Specifically, to insert $\row{p}$ of vertices below $\row{p_a}$ and above $\row{p_b}$ in row $j$, let the critical point $p$ lie on the $i'$th vertical boundary. For $1 \leq i \leq n$, if $i = i'$, then set $\weight{\vertex{\pi_{i'}, H^r_j[1]}} = 1$. For every other value of $i$, set $\weight{\vertex{\pi_i}, H^r_j[1]} = 1$, if the weights of both $\vertex{\pi_i, p_a}$ and $\vertex{\pi_i, p_b}$ are 1. Otherwise, set $\weight{\vertex{\pi_i}, H^r_i[1]} = 0$. Reduce the size of $H^r_j$ by one by removing $H^r_j[1]$. 
    
    We prove the correctness of the boundary vertex operation by showing that this insertion process maintains the properties of the grid graph. First, the total number of row critical points plus the placeholder points stays the same. Second, it is sufficient to determine the weight of $\vertex{\pi_i, p}$ by checking the weights of $\vertex{\pi_i, p_a}$ and $\vertex{\pi_i, p_a}$. Both intersections $l(\pi_i) \cap l(p_a)$ and $l(\pi_i) \cap l(p_b)$ need to lie in the freespace for $l(\pi_i) \cap l(p)$ to lie in the freespace, since the opposite suggests that there is a grid line between $l(p_a)$ and $l(p_b)$, contradicting the assumption that $l(p_a)$ and $l(p_b)$ are adjacent.
    
    If $u$ is a row operation, to delete $\row{p}$ lying in the $j$th row, let $l(q)$ be the first grid line below $l(H^r_j[1])$. For all $1 \leq i \leq n$, set $\weight{\vertex{\pi_i, q}} = \weight{\vertex{\pi_i, \sigma_j}}$. Insert a new placeholder point at the beginning of $H^r_j$ by setting $H^r_j[1] = q$. This process also maintains the properties of the grid graph. Column operations uses analogous arguments. 
\end{proof}

Given $\ggraph_{i - 1} = \ggraph_\delta(\pi, \sigma + t_{i - 1})$ and event $t_{i}$, we can now transform $\ggraph_{i - 1}$ to $\ggraph_{i}$. Specifically, in Lemma~\ref{lem:fsgraph-event-updates}, we have shown that if $t_{i}$ is a VE event, it takes $\Oh(n)$ time plus $\Oh(1)$ corner vertex operations or row operations. If $t_i$ is a VVE event, it takes $\Oh(1)$ time plus $\Oh(1)$ boundary vertex operations. By combining Lemma~\ref{lem:fsgraph-event-updates} and Lemma~\ref{lem:ggraph-number-of-update-per-op}, we can bound the number of vertex weight changes for each freespace event type. 
\begin{lemma} \label{lem:event-to-weight-changes}
    Given $\ggraph_\delta(\pi, \sigma + t_{i - 1})$ and the next event $t_i$, to compute $\ggraph_\delta(\pi, \sigma + t_{i})$, it takes 
    \begin{itemize}
        \item $\Oh(n)$ vertex weight changes if $t_i$ is a VE event, or
        \item $\Oh(1)$ vertex weight changes if $t_i$ is a VVE event. 
    \end{itemize}
\end{lemma}

We can now summarize Sections~\ref{sec:fs-reach-to-fsg-reach} and~\ref{sec:fsg-to-gg} and state the main lemma of this section. In Lemma~\ref{lem:fsgraph-path}, we have shown that the Fréchet distance $\df{\pi, \sigma}$ is at most $\delta$ if and only if the refined freespace graph $\fsgraph = \fsgraph_\delta(\pi, \sigma)$ is $st$-reachable. In Section~\ref{sec:fsg-to-gg}, for each $\fsgraph$, we have defined an associate grid graph $\ggraph = \ggraph_\delta(\pi, \sigma)$. In Lemma~\ref{lem:fs-gg-path-equivalent}, we have shown that $\ggraph$ is $st$-reachable if and only if $\fsgraph$ is $st$-reachable. Combining the above with Lemma~\ref{lem:event-to-weight-changes}, we have the following.

\begin{lemma} \label{lem:ggraph-summarized}
    Let $T = \{t_1, ..., t_m\}$ be a complete set of freespace events containing exclusively $m_{vve}$ VVE events and $m_{ve}$ VE events. Let $N = 2n^2 - n$, and let $T_u(N)$ (resp. $T_q(N)$) be the time complexity to update (resp. query $st$-reachability) in an $N \times N$ grid graph. It takes
    \begin{align*}
        \Oh(N^2 + (m_{ve} \cdot n + m_{vve}) \cdot T_u(N) + m \cdot T_q(N))
    \end{align*}
    time to determine if there exists $t_i \in T$ such that $\df{\pi, \sigma + t_i} \leq \delta$. 
\end{lemma}

Using the results by Alt, Knauer and Wenk~\cite{altMatchingPolygonalCurves2001}, we can build an arrangement in the translation space. Using this arrangement, we can compute a set of complete events (translations) $T$ containing exclusively $\Oh(n^6)$ VVE events and $\Oh(n^5)$ VE events. They have shown that it is sufficient to consider only translations in $T$ to determine if there is a translation $t$ such that $\df{\pi, \sigma + t} \leq \delta$. 

We analyze the running time. In total, we require $\Oh(n^6)$ vertex updates for $\ggraph$. Using the result of Bringmann, K\"unnemann and Nusser~\cite{bringmannFrechetDistanceTranslation2021}, we can update a vertex and perform $st$-reachability queries in amortized $\Oh(N^{2/3} \cdot \log^2 N) = \Oh(n^{4/3} \cdot \log^2 n)$ time. We obtain the following theorem. We defer the proof of Theorem~\ref{theorem:seven-point-threethree}, and its generalization, to Section~\ref{sec:together}. 
\begin{theorem}
    \label{theorem:seven-point-threethree}
    The Fréchet distance under translation in $\reals^2$ can be decided in $\Oh(n^{7+ \frac{1}{3}} \log^2 n)$ time. 
\end{theorem}

\section{Fréchet distance under transformation} \label{sec:together}

In this section, we consider a class $\trans$ of transformations that is \emph{rationally parameterized} or \emph{rationally represented with $k$ degrees of freedom} as defined by Wenk~\cite{wenkShapeMatchingHigher2003}.

\begin{definition}[{{\cite[Definition~25]{wenkShapeMatchingHigher2003}}}] \label{def:para-trans}
    Let $1 \leq k \leq d^2 + d$, and let $p_1, \dots, p_{d^2+d}, q_1, \dots, q_{d^2+d} \in \mathbb{R}[X_1, \dots, X_k]$ be $2(d^2 + d)$ polynomials of constant degree in $k$ variables, such that $q_i(x) \neq 0$ for all $1 \leq i \leq d^2 + d$ and for all $x \in \mathbb{R}^k$. Let $r_i := p_i / q_i$ for all $1 \leq i \leq d^2 + d$, such that $r_i(x) := p_i(x) / q_i(x)$ for all $x \in \mathbb{R}^k$. If 
        \[
        \mathcal{T} = 
        \left\{
        \begin{pmatrix}
        r_1(x) & \dots & r_d(x) \\
        r_{d+1}(x) & \dots & r_{2d}(x) \\
        \vdots & \ddots & \vdots \\
        r_{d^2-d+1}(x) & \dots & r_{d^2}(x)
        \end{pmatrix}, 
        \begin{pmatrix}
            r_{d^2 + 1}(x) \\
            r_{d^2 + 2}(x) \\
            \vdots \\
            r_{d^2 + d}(x)
        \end{pmatrix}
        \Bigg| \; x \in \mathbb{R}^k
        \right\},
        \]
    then we call $\mathcal{T}$ \textit{rationally parameterized} or \textit{rationally represented with $k$ degrees of freedom} (dof). $\mathbb{R}^k$ is called the parameter space of $\mathcal{T}$.
\end{definition}

Let $\reals^k$ be the parameter space of $\trans$. For $x \in \reals^k$, let $\tau_x$ denote the transformation defined by the $k$-tuple $x$ of parameters. Let $\pi$ and $\sigma$ be two $d$-dimensional polygonal curves. Let $\tau_x(\sigma)$ be the resulting curve by applying the transformation $\tau_x$ to $\sigma$. 

In $\reals^k$, a \emph{vertex-vertex-edge (VVE) critical transformation} $\vvetrans(\pi_i, \pi_j, \cseg{\sigma_w})$ is the union of every point $x \in \reals^k$ such that the segment $\tau_x(\cseg{\sigma_w})$ lies on the intersection of the boundary of the $d$-spheres centered at $\pi_i$ and $\pi_j$, respectively, and it is formally defined as follows. 
\begin{align*}
    \vvetrans(\pi_i, \pi_j, \cseg{\sigma_w}) = \{x \in \reals^k \mid \exists z \in \cseg{\sigma_w}, \abse{\tau_x(z) - \pi_i} = \abse{\tau_x(z) - \pi_j} = \delta \}
\end{align*}

Analogously, a \emph{vertex-edge (VE) critical transformation} $\vetrans(\pi_i, \cseg{\sigma_w})$ is the union of every point $x \in \reals^k$ such that the segment $\tau_x(\cseg{\sigma_w})$ lie on the boundary of the $d$-ball centered at $\pi_i$, and it is formally defined as follows. 
\begin{align*}
    \vetrans(\pi_i, \cseg{\sigma_w}) = \{x \in \reals^k \mid \exists z \in \cseg{\sigma_w}, \abse{\tau_x(z) - \pi_i} = \delta\}
\end{align*}

\newcommand{\arrangement}[0]{\mathcal{A_\delta}}
Every critical transformation is a semi-algebraic set. Using $\Oh(n^3)$ VVE critical transformations and $\Oh(n^2)$ VE critical transformations, Wenk~\cite[Proof of Theorem~8]{wenkShapeMatchingHigher2003} showed that they can build an arrangement $\arrangement$ in $\Oh(n^{3k})$ time using $\Oh(n^{3k})$ space. Furthermore, $\arrangement$ contains at most $\Oh(n^{3k})$ $k'$-dimensional faces for $0 \leq k' \leq k - 1$. Let $\tau_F = \tau_x$ be the transformation that is represented by an arbitrary parameter $x \in F$. In \cite[Lemma~24]{wenkShapeMatchingHigher2003}, Wenk~showed that if there exists some transformation $\tau$ such that $\df{\pi, \tau(\sigma)} < \delta$, then there exists some $k'$-dimensional face $F \in \arrangement$ such that for any $\tau_F$, $\df{\pi, \tau_F(\sigma)} \leq \delta$. Their results are summarized as follows. 
\begin{fact} \label{fac:wenk-summarized}
    Given a pair of polygonal curve $\pi$ and $\sigma$, a real number $\delta \geq 0$, and a class $\trans$ of transformations that is rationally represented with $k$ degrees of freedom, one can build an arrangement $\arrangement = \arrangement(\pi, \sigma, \trans)$ using at most $\Oh(n^3)$ VVE critical transformations and $\Oh(n^2)$ VE critical transformations. The arrangement $\arrangement$ has in total $\Oh(n^{3k})$ complexity, and it can be constructed in $\Oh(n^{3k})$ time using $\Oh(n^{3k})$ space. To determine if there exists a transformation $\tau \in \trans$ such that $\df{\pi, \tau(\sigma)} \leq \delta$, it is sufficient to check exactly one transformation $\tau_F$ for every $k'$-dimensional face $F \in \arrangement$, where $0 \leq k ' \leq k - 1$. 
\end{fact}

Once the arrangement~$\arrangement$ is constructed, the remainder of the previous algorithm in Wenk~\cite{wenkShapeMatchingHigher2003} is straightforward. For every face $F \in \arrangement$, sample a point $x \in F$, and determine if $\df{\pi, \tau_x(\sigma)} \leq \delta$ using classic algorithms (Alt and Godau~\cite{altComputingFrechetDistance1995} for example). In total, this takes $\tOh(n^{3k} \cdot n^2)$ time. 

To obtain a running time improvement, we use a similar approach to previous sections. We generate a complete set of events as follows. Initialize an empty graph $\graph = (\vertices, \edges)$. For every face $F \in \arrangement$, add a vertex $v_F$ to $\vertices$. For every two adjacent faces $F$ and $F'$, add an edge $(v_F, v_{F'})$ to $\edges$. For each vertex $v_F$, record a transformation $\tau_F$. Next, we compute a complete set of events using $\graph$. Initialize an empty set of events~$T$. Then, use a DFS to compute a spanning tree of $\graph$, and perform a Euler tour over the spanning tree starting from an arbitrary vertex. For each directed edge $e = (v_{F'}, v_{F})$ in the tour, we add an event only if we enter or leave a critical transformation. More specifically, let $B(F)$ be the set of critical transformations adjacent to face $F$. If $B(F) \setminus B(F') = \{\trans_{c}\}$, we say $e$ \emph{traverses onto} the critical transformation $\trans_{c}$ via $F$. If $B(F') \setminus B(F) = \{\trans_c\}$, we say $e$ \emph{traverses out of} the critical transformation $\trans_{c}$ via $F$. 

Let $\disk{p}$ be the $d$-sphere of radius $\delta$ centered at $p$. Depending on the cases where $e$ traverse onto or out of a VVE or VE critical transformation, we add the respective freespace events defined in Definition~\ref{def:fs-events}. 
\begin{enumerate}
    \item If $\trans_c$ is a VVE critical transformation $\vvetrans = \vvetrans(\pi_i, \pi_j, \cseg{\sigma_w})$, we compute $p = \disk{\pi_i} \cap \tau_F(\cseg{\sigma_w})$ and $q = \disk{\pi_j} \cap \tau_F(\cseg{\sigma_w})$, and append an event $t$ defined by $(p, q, \cseg{\sigma_w})$ to $T$. If $e$ traverses onto $\trans_c$, $t$ is an overlapping event. If $e$ traverses out of $\vvetrans$, $t$ is a separating event. 
    \item If $\trans_c$ is a VE critical transformation $\vetrans = \vetrans(\pi_i, \cseg{\sigma_w})$, we compute $p = \disk{\pi_i} \cap \cseg{\sigma_w}$, and we append an event $t$ represented by $(p, \cseg{\sigma_w})$ to $T$. If $e$ traverses onto $\trans_c$, $t$ is an entering event if $p = \pi_i$, or an appearing event if $p \neq \pi_i$. Analogously, if $e$ traverses out of $\trans_c$, $t_e$ is a leaving event if $p = \pi_i$, or a disappearing event if $p \neq \pi_i$. 
\end{enumerate}

We say an event $t$ is associated with the critical transformation $\trans_c$ and the edge $e$, if $t$ is computed from an edge $e$ traversing onto or out of $\trans_c$. We show that $T$ has two desired property.

\begin{lemma} \label{lem:events-from-arrangement}
    Given the arrangement $\arrangement = \arrangement(\pi, \sigma, \trans)$, one can compute a complete set $T = \{t_1, ..., t_{\Oh(n^{3k})}\}$ of freespace events in $\Oh(n^{3k})$ time with the following properties.
    \begin{enumerate}
        \item Every face in $\arrangement$ is associated with at least one event in $T$. 
        \item For each event $t_i$ associated with edge $(v_F, v_{F'}) \in \edges$, $\fsgraph_\delta(\pi, \tau_F(\sigma))$ and $\fsgraph_\delta(\pi, \tau_{F'}(\sigma))$ differ by exactly one freespace event. 
    \end{enumerate}
\end{lemma}

\begin{proof}
    First, observe that each edge $(v_F, v_{F'})$ in the Euler tour traverses onto or out of at most one critical transformation. The opposite suggests that the faces $F$ and $F'$ are not adjacent. It is also clear that $\graph$ is a connected component containing a vertex $v_F$ for every face $F \in \arrangement$, since the parameter space $\reals^k$ itself is also a face. An Euler tour over the spanning tree visits every vertex, and hence every face at least once. 

    We next argue that the freespace graph does not change unless a freespace event occurs. It is clear that unless an appearing or disappearing event occurs, neither the vertices nor the edges in a freespace graph $\fsgraph$ change. What remains is to argue that the vertex weights do not change unless an event occurs. More specifically, unless a freespace event occurs, no intersection $a = l(p) \cap l(q)$ enter or leave the freespace. 

    For the sake of contradiction, say that $a$ either enters or leaves the freespace and a freespace event does not occur. Clearly, $\vertex{p, q}$ cannot be a corner vertex, as the weight change is explicitly captured by the entering/leaving event. $\vertex{p, q}$ cannot be an interior vertex, as every interior vertex has weight $1$ regardless. 
    
    Therefore, $\vertex{p, q}$ is a boundary vertex and $a$ is a critical point. Without loss of generality, let $p$ be a vertex of $\pi$. If $a$ enters the freespace, $a$ must coincide with a critical point $b = l(p') \cap l(q)$, so grid lines $l(p)$ and $l(p')$ overlap. If $a$ leaves the freespace, $a$ must first coincide with (again say) $b$, so the grid lines $l(p)$ and $l(p')$ separate. In both cases, either an overlapping event occurs or a separating event occurs, contradicting the assumption. 

    Let $p$ (resp. $q$) be a critical point on the $i$th (resp. $j$th) FSD boundary. We next argue that no freespace event occurs unless we traverses onto or out of a critical transformation. This is true by the definition. An entering/appearing (resp. leaving/disappearing) event defined by $(p, \cseg{\pi_w})$ occurs only when the associated edge $(v_F, v_{F'})$ traverses onto (resp. out of) the critical transformation $\vetrans(\pi_i, \cseg{\sigma_w})$. An overlapping (resp. separating) event defined by $(p, q, \cseg{\sigma_w})$ occurs only when $(v_F, v_{F'})$ traverses onto (resp. out of) the critical transformation $\vvetrans(\pi_i, \pi_j, \cseg{\sigma_w})$.  
    
    Observe that exactly one freespace event occurs every time we traverses onto or out of a critical transformation, which is captured by the Euler tour. Therefore, every face is associated with at least one event in $T$, and for every edge $(v_F, v_{F'})$ associated with an event $t \in T$, $\fsgraph(\pi, \tau_F(\sigma))$ and $\fsgraph(\pi, \tau_{F'}(\sigma))$ differ by exactly one event. 
\end{proof}

We also observe that fewer faces are adjacent to VE critical transformations than to VVE critical transformations. 
\begin{restatable}{observation}{numVEfaces} \label{obs:number-vve-ve-faces}
    In the arrangement $\arrangement = \arrangement(\pi, \sigma, \trans)$, there are at most $\Oh(n^{3k})$ faces adjacent to VVE critical transformations, and $\Oh(n^{3k - 1})$ faces adjacent to VE critical transformations. 
\end{restatable}
\begin{proof}
    Wenk~\cite{wenkShapeMatchingHigher2003} proved that there are at most $\Oh(n^{3k})$ faces in $\arrangement$, which upperbounds the number of faces adjacent to VVE critical transformations. They also showed that a critical transformation is a semi-algebraic set with constant description complexity~\cite[Lemma 24]{wenkShapeMatchingHigher2003}. Therefore, each face in $\arrangement$ is the intersection of at most $k$ critical transformations. If a face is adjacent to a VE critical transformation, there are at most $n^{3(k - 1)}$ possible ways to choose $k - 1$ from $n^3$ critical transformations. There are at most $n^2$ VE critical transformations, and in total, $\Oh(n^{3k - 1})$ faces are adjacent to VE critical transformations. 
\end{proof}

To obtain fast updates and queries in the grid graph, we use the offline dynamic grid reachability result by~\cite{bringmannFrechetDistanceTranslation2021}. The problem is defined as follows. We start from a directed $N \times N$-grid graph, and we are given a set $\{u_1, ..., u_U\}$ of updates such that each update $u_i$ is to either activate or deactivate a vertex. For $1 \leq i \leq U$, the goal is to compute after each update $u_i$ whether there is a feasible path from vertex $\vertex{1, 1}$ to $\vertex{N, N}$. Their result is as follows. 
\begin{fact}[Frechet distance under translation{\cite[Theorem~3.4]{bringmannFrechetDistanceTranslation2021}}] \label{fac:bkn-ggraph-theorem}
    Offline dynamic grid reachability can be solved in time $\Oh(N^2 + U N^{2/3} \log^2 N)$. 
\end{fact}

We summarize the results in this paper, which we will then combine to obtain a faster algorithm for computing the Fr\'echet distance under rationally parameterized transformations. We take as inputs a real number $\delta \geq 0$, a class $\trans$ of transformations rationally represented with $k$ degrees of freedom, and a pair of polygonal curves $\pi$ and $\sigma$, each with $n$ vertices. Our goal is to determine if there exists a transformation $\tau \in \trans$ such that $\df{\pi, \tau(\sigma)} \leq \delta$. 

To do this, Wenk~\cite{wenkShapeMatchingHigher2003} has shown that it is sufficient to construct an arrangement $\arrangement$ in the parameter space $\reals^k$ and sample a transformation from each face. In Lemma~\ref{lem:events-from-arrangement}, we have shown that we can traverse this arrangement and generate a complete set of $\Oh(n^{3k})$ freespace events in $\Oh(n^{3k})$ time. In Section~\ref{sec:fs-reach-to-fsg-reach} and~\ref{sec:fsg-to-gg}, we explained our ideas using translation. However, all results revolve around handling a complete set of freespace events, and hence apply to more general settings of transformations. We can therefore plug in the number of updates to Lemma~\ref{lem:ggraph-summarized}. Our algorithm takes $\Oh(n^4 + n^{3k} \cdot (T_u(n^2) + T_q(n^2))$ time, where $T_u(n^2)$ (resp. $T_q(n^2)$) is the time complexity to update a vertex (resp. query $st$-reachability) in an $\Oh(n^2) \times \Oh(n^2)$-grid graph. By Fact~\ref{fac:bkn-ggraph-theorem}, $T_u(n^2) + T_q(n^2)$ takes amortized $\Oh(n^{4/3} \log^2 n)$ time. For classes of rationally parameterized transformations with at least one degree of freedom, $k \geq 1$. 
We state our final result.
\begin{theorem}
    The Fréchet distance under transformations rationally represented with $k$ degrees of freedom can be decided in $\Oh(n^{3k + 4/3} \log^2 n)$ time. 
\end{theorem}

\section{Conclusion}

Our algorithm provides the first progress in over 20 years for computing the (continuous) Fr\'echet distance under transformations. The running time comes from traversing an arrangement in the space of transformations, performing an update to a dynamic grid graph data structure in each step. 

We conclude with open questions. Can the update time be reduced? Improving the update time for the data structure of~\cite{bringmannFrechetDistanceTranslation2021} would directly improve our running time. But it may also be possible to tailor the data structure for our setting. For instance, we could prune the lower levels of the data structure, since reachability in the inside of free space cells does not carry any information.

Can we prove a non-trivial conditional lower bound for computing the Fr\'echet distance under translations? The complexity of the arrangement, $\Omega(n^6)$, would seem like a natural lower bound. However, even transferring the $n^{4-o(1)}$ conditional lower bound for the discrete Fr\'echet distance under translation~\cite{bringmannFrechetDistanceTranslation2021} to the (continuous) Fr\'echet distance seems difficult, since the lower bound construction crucially relies on the fact that the discrete Fréchet traversal can only stay on the vertices and not on the edges.

\bibliographystyle{plainurl}
\bibliography{references} 

\end{document}